\documentclass[11pt]{amsart}
\usepackage{amsmath,amstext,amsbsy,amssymb}

\newtheorem{theorem}{Theorem}[section]    
\newtheorem{lemma}[theorem]{Lemma}
\newtheorem{proposition}[theorem]{Proposition}
\newtheorem{corollary}[theorem]{Corollary}

\newcommand{\la}{\lambda}

\newcommand{\ep}{\epsilon}

\begin{document}
\title{QUANTUM ENTANGLEMENT AND APPROXIMATION BY POSITIVE MATRICES}
\author{Xiaofen Huang}
\address{School of Mathematics and Statistics,
Hainan Normal University,Haikou, Hainan 571158, China}
\email{huangxf1206@googlemail.com}
\author{Naihuan Jing*}
\address{School of Sciences, South China University of Technology, Guangzhou 510640, China
and Department of Mathematics, North Carolina State University, Raleigh, NC 27695, USA}
\email{jing@math.ncsu.edu}
\thanks{*Corresponding author}
\keywords{Quantum computation; density matrices; nonlinear
optimization.}

\begin{abstract}
We give an exact solution to the nonlinear optimization problem of
approximating a Hermitian matrix by positive semi-definite
matrices. Our algorithm was then used to judge whether a quantum
state is entangled or not. We show that the exact approximation of a
density matrices by tensor product of positive semi-definite
operators is determined by the additivity property of the density
matrix.
\end{abstract}
\maketitle


\section{Introduction}\label{aba:sec1}
Quantum entanglement is one of the most interesting features in
quantum computation \cite{NC, P} as it is vital for quantum dense
coding, quantum error-correcting codes, teleportation and also
responsible for fast quantum algorithms.   Much efforts have been
made to find criteria for quantum separability: the Bell
inequalities \cite{B}, PPT (positive partial transposition)
\cite{P}, reduction criterion \cite{CAG, HH}, majorization criterion
\cite{NK}, entanglement witnesses \cite{HLVC, T, LKCH}, realignment
\cite{R, CW} and generalized realignment methods \cite{ACFH}.
Nevertheless the problem is still not fully solved except for lower
rank cases \cite{HLVC, AFG, FGW}.

In \cite{FJS} we proposed a new method to judge if a density matrix
is separable. We first decompose the density matrix as a tensor
product of hermitian matrices, and then we reduce the separability
problem to that of finding when the hermitian matrix becomes
positive semi-definite. The strategy was to solve the separability
by two steps: first one finds a tensor decomposition of the density
matrix by hermitian matrices, and then one approximates the
hermitian matrices by positive semi-definite matrices if possible.
Although it was proved that a density matrix is separable if and
only if the separability indicator is non-negative, it is highly
nontrivial to actually compute this indicator. In this sense the
method is also similar to many of its predecessors.

In the current paper we approach the question from a new angle by
giving an algorithm to compute the closest positive semi-definite
matrix to any given hermitian matrix. It is noted that this
optimization is not a linear problem so the usual QR decomposition
in \cite{HJn} does not work. By general theory of convex sets the
existence of the minimum is guaranteed but it is nontrivial to find
the exact solution due to the nonlinearity. In this paper we first
solve this optimization problem exactly using matrix theory. This
paves the way for us to attack the main problem of separability by
directly looking for an optimal approximation to a sum of matrices
by positive semi-definite matrices.

It turns out that the exact solution in the most important case can
be solved by Lie theoretic techniques. First we show that a
Hermitian matrix with two commuting summands can be approximated
term by term. Next we prove that if the summands of a Hermitian
matrix can be simultaneously made to upper triangular matrices, then
the exact approximation by positive semi-definite matrix can be done
term by term. Our method provides a direct way to approximate the
hermitian matrix by positive semi-definite matrices and thus the
separability problem can be solved theoretically in this sense.

It is interesting to note that similar (but stronger constrained)
matrix approximation also appears in finance, image processing, date
mining, and other areas such as resource allocation and industrial
process monitoring \cite{SH, ZW, BW}. Most methods used in these
problems are numerical algorithms that only give an approximation to
the solution. In some sense our results also provide the first exact
and analytical method in this direction, and we also hope that
conversely some of the numerical algorithms may be useful in quantum
entanglements.

\section{Approximation by Positive Definite Matrices}
Let $A$ be an $n\times n$ Hermitian matrix, and let $Q$ be a unitary
matrix such that $A=QDQ^{\dagger}$, where $D=diag(\alpha_1, \cdots,
\alpha_n)$ and $\dagger$ is conjugation and transposition. The
signature $(p, q)$ of $A$ (cf. \cite{L}) is defined by
$p+q=rank(A)$. We can permute the columns of the matrix $Q$ so that the eigenvalues
of $A$ are arranged in the following order:
$$
\alpha_1\geq \cdots\geq\alpha_p>0>\alpha_{p+1}\geq \cdots \geq
\alpha_{p+ q}, \alpha_{p+q+1}=\cdots =\alpha_n=0.
$$
We further define
\begin{equation} \label{posdecomp}
A=A_+-A_-,
\end{equation}
where $A_{\pm}=QD_{\pm}Q^{\dagger}$, and
\begin{align*}
D_+&=diag(\alpha_1, \cdots, \alpha_p, 0, \cdots, 0)\\
D_-&=diag(0, \cdots, 0, -\alpha_{p+1}, \cdots, -\alpha_{p+q}, 0, \cdots, 0)
\end{align*}
formed by positive (negative) eigenvalues respectively. We remark
that our definition of the positive and negative semi-definite parts
of $A$ is independent from our choice of $Q$. In general, if $A$ is
diagonalized by a unitary matrix $Q$ as follows:
$$A=QDQ^{\dagger}=Qdiag(\alpha_1,
\cdots, \alpha_n)Q^{\dagger},
$$
then we have
$$A_{\pm}=Qdiag(\alpha_1^{\pm},
\cdots, \alpha_n^{\pm})Q^{\dagger}, \qquad
\alpha_i^{\pm}=\frac{|\alpha_i|\pm \alpha_i}2.
$$

It is clear that $A$ is positive semi-definite if and only all
eigenvalues of $A$ are non-negative. Therefore both $A_{\pm}$ are
positive semi-definite matrices. It is easy to see that the
decomposition (\ref{posdecomp}) of $A$ into a difference of positive
semi-definite matrices is unique up to positive definite matrices,
i.e., if $A=A_+'-A_-'$, where $A_{\pm}'$ are positive semi-definite,
then $A_{\pm}=A_{\pm}'+P$ with a positive semi-definite matrix $P$.

If $A$ and $B$ are commuting positive semi-definite matrices, so is
their product. If they are not commutative, then $AB$ is in general
not a positive semi-definite matrix, as $AB$ may not even be
hermitian.

\begin{lemma}\label{prod} Let A and B be positive semi-definite Hermitian matrices, then
the eigenvalues of $AB$ are all non-negative.
\end{lemma}
\begin{proof} If $A$ is invertible, and let $\la$ be an eigenvalue
of $AB$, then for some vector $x\neq 0$,
$$ABx=\la x.$$
Hence $Bx=\la A^{-1}x$. Note that both $A^{-1}$ and $B$ are
positive, so
$$x^{\dagger}Bx=\la x^{\dagger}A^{-1}x>0,$$
Subsequently $\la>0$. In general, if $A$ is singular, then for any
$\epsilon>0$ the matrix $A+\ep I$ is positive definite. Therefore
any eigenvalue $\la=\la(\ep)$ of $(A+\ep I)B$ is positive. Letting
$\ep\rightarrow 0$, we see that $lim_{\ep\to 0}\la \geq 0$, i.e. any
eigenvalue of $AB$ is non-negative.
\end{proof}

For any matrix $A$, the Frobenius norm is defined to be $||A||_F$
$=(tr(AA^{\dagger}))^{1/2}$, which is also equal to the sum of
squares of singular values of $A$.

\begin{theorem} Let $A$ be an $n\times n$ hermitian matrix,
then for any positive semi-definite matrix $B$ we have
\begin{equation}
||A-B||_F\geq ||A_-||_F
\end{equation}
with equality when $B=A_+$. i.e. the closest positive semi-definite
matrix to $A$ is given by $A_+$.
\end{theorem}
\begin{proof} Omitting the subscript
in the Frobenius norm, we have for any positive semi-definite matrix
$B$
\begin{align*}
\|A-B\|^2&=tr(A-B)^2=tr(A^2)-2tr(AB)+tr(B^2)\\
&=tr(A^2)+2tr(A_-B)+tr(B^2)-2tr(A_+B)\\
&=tr(A^2)+2tr(A_-B)+tr(B-A_+)^2-tr(A_+^2)
\end{align*}
by using $tr(AB)=tr(BA)$ and completing square. Since $B$ and
$A_{\pm}$ are positive semi-definite, then $tr(A_-B)\geq 0$ for any
$B$ by Lemma \ref{prod}. Therefore, it follows that for any positive
semi-definite matrix $B$
\begin{align*}
\|A-B\|^2&\geq tr(A^2)+tr(B-A_+)^2-tr(A_+^2)\\
&=\|A\|^2-\|A_+\|^2+\|B-A_+\|^2\\
&\geq \|A\|^2-\|A_+\|^2
\end{align*}
where the equality is obtained when $B=A_+$.
\end{proof}

We remark that similar (Toeplitz and/or correlation matrix)
approximation with stronger constraints has been studied in finance
and image processing \cite{SH, ZW}. Our result is more general and
stronger in the sense that we do not require that the matrix to be
either Toeplitz or correlation matrix (real positive definite with
unit diagonal). Furthermore our result is analytical and exact, as
no numerical approximation is needed for the solution.

\begin{corollary} Let A and B be any two Hermitian matrices,
then for any positive semi-definite matrix C of the same size as
$A\otimes B$, we have
$$\|A\otimes B-C \|_{F}\geq\|(A \otimes B) _{-}\|_F=\|A_+\otimes B_{-}+A_{-}\otimes B_{+}\|_F.$$
The equality holds when $C=(A \otimes B)_+=A_+ \otimes B_+ +A_-
\otimes B_-$.
\end{corollary}

\begin{proof} It is enough to show that $(A\otimes
B)_{\pm}=\sum_{\ep}A_{\ep}\otimes B_{\pm\ep}$, where $\ep={+, -}$.
Suppose $A$ and $B$ are diagonalized by $Q_1$ and $Q_2$
respectively:
\begin{align*}
A&=Q_1D_1Q_1^{\dagger}=Q_1diag(\alpha_1, \cdots, \alpha_m)Q_1^{\dagger}, \\
B&=Q_2D_2Q_2^{\dagger}=Q_2diag(\beta_1, \cdots,
\beta_n)Q_2^{\dagger}
\end{align*}
then we have
$$A\otimes B=(Q_1\otimes Q_2)(D_1\otimes D_2)(Q_1\otimes Q_2)^{\dagger} $$
As the eigenvalues $\sigma$ of $A\otimes B$ are $\la_i\beta_j$, thus
\begin{align*}
\sigma_{ij}^+&=\alpha_i^+\beta_j^++\alpha_i^-\beta_j^-, \\
\sigma_{ij}^-&=\alpha_i^+\beta_j^-+\alpha_i^-\beta_j^+, \\
\sigma_{ij}^{0}&=\alpha_i^0\beta_j+\alpha_i\beta_j^0.
\end{align*}
Since the zero eigenvalues do not contribute to the decomposition,
we have that  $(A\otimes B)_+=A_+\otimes B_+ +A_-\otimes B_-$, and
$(A\otimes B)_-=A_+\otimes B_-+A_-\otimes B_+$
\end{proof}

\section{Sums of matrices and estimates}

It appears that the decomposition of the Hermitian matrix $A$ into a
sum of two Hermitian matrices $B, C$ has a close relationship to our
problem. This problem has a much longer history in mathematics.

Horn \cite{Ho} defines the following concept for the Hermitian
matrices. Let $A$ and $B$ be  two $n\times n$ matrices with
eigenvalues $\alpha_i$ and $\beta_i$. For any subset $I$ we denote
$|I|=\sum_{i\in I} i$. The set $T_r^n$ of triples $(I, J, K)$ of
subsets of $\{1, \ldots, n\}$ of the same cardinality $r$ is defined
by first setting
\begin{equation*}
U_r^n=\{(I, J, K)| |I|+|J|=|K|+r(r+1)/2\},
\end{equation*}
then define $T_1^n=U_1^n$ and in general
\begin{align*}
T_r^n=\{(I, J, K)&\in U_r^n| \mbox{\, for all $p<r$ and all\, } (F,
G,
H)\in U_p^n \\
&  \sum_{f\in F}i_f+\sum_{g\in G}j_g\leq\sum_{h\in H}k_h+p(p+1)/2\},
\end{align*}

The following characterization of eigenvalues $\alpha, \beta,
\gamma$ of $C=A+B$ is proved in \cite{F}.

\begin{theorem}({\bf Horn's conjecture}) A triple $(\alpha, \beta,
\gamma)$ occurs as eigenvalues of $A, B, C$ such that $C=A+B$ if and
only if $|\alpha|+|\beta|=|\gamma|$ and inequalities
\begin{equation*}
\sum_{k\in K}\gamma_k\leq\sum_{i\in I}\alpha_i+\sum_{j\in J}\beta_j
\end{equation*}
hold for all triple $(I, J, K)$ in $T_r^n$ for all $r<n$.
\end{theorem}
The special case of $r=1$ is Weyl's inequality:
\begin{equation*}
\gamma_{i+j-1}\leq\alpha_i+\beta_j
\end{equation*}
Since $|\gamma|=|\alpha|+|\beta|$, we also have
\begin{equation*}
\sum_{k\in K^c}\gamma_k\geq\sum_{i\in I^c}\alpha_i+\sum_{j\in
J^c}\beta_j
\end{equation*}
A practical bound is the following:
\begin{equation*}\displaystyle
\underset{i+j=n+k}{\mbox{Max}}\alpha_i+\beta_j\leq\gamma_k\leq
\underset{i+j=k+1}{\mbox{Min}}\alpha_i+\beta_j
\end{equation*}

Apply these results to our situation, we then get the following:
\begin{theorem} Let $A$ be a density matrix over the Hilbert space
$\mathcal H=\mathcal H_1\otimes \mathcal H_2$. Suppose $A=\sum_i
B_i\otimes C_i$, then we have
\begin{equation*}
||A-A_+||_F\leq \sum_i||B_i-(B_i)_+|| \cdot||C_i-(C_i)_+||
\end{equation*}
\end{theorem}

\section{Lie Algebras and Approximation of Summations}

As the Hermitian decomposition of the density matrix $A$ is a
summation, one hopes to check its separability by demonstrating that
each summand can be expressed by tensor products of positive
semi-definite
 matrices.

 The main problem is to estimate the error given by term by term
 approximation. Suppose that $A$ and $B$ are two Hermitian
 matrices we would like to estimate the
 norm $||A+B-A_+-B_+||$.

\begin{lemma}Let A and B are two commuting Hermitian matrices, then
\begin{align*}
(A+B)_+&\leq A_+ +B_+\\
(A+B)_-&\leq A_- +B_-
\end{align*}
\end{lemma}
\begin{proof} Since $A$ and $B$ are commuting, they can be diagonalized
simultaneously by a unitary matrix $Q$. In other words we have
\begin{align*}
A&=QD_AQ^{\dagger}=A_+ -A_- =Q(D_A)_+Q^{\dagger} -Q(D_A)_-Q^{\dagger}\\
B&=QD_BQ^{\dagger}=B_+ -B_- =Q(D_B)_+Q^{\dagger}
-Q(D_B)_-Q^{\dagger}
\end{align*}
thus,
\begin{align*}
A+B&=Q(D_A+Q_B)Q^{\dagger}\\
&=Q((D_A)_++(D_B)_+)Q^{\dagger} -Q((D_A)_-+(D_B)_-)Q^{\dagger},
\end{align*}
which implies that $(A+B)_+=A_++B_+$ and $(A+B)_-=A_-+B_-$.
\end{proof}

\begin{theorem} \label{T:factor}
Let $A_i$ (resp. $B_i$) be set of commuting positive semi-definite
Hermitian matrices of the same size, then for any positive definite
matrix $C$ of the same size as $A_i\otimes B_i$ we have
$$
\|\sum_i A_i \otimes B_i-C\|_F\geq\|\sum_i \left\{(A_i)_+ \otimes
(B_i)_-+(A_i)_- \otimes (B_i)_+\right\}\|_F,
$$
where the equality holds when
$$C=\sum_i(A_i\otimes B_i)_+=\sum_i
\left\{(A_i)_+ \otimes (B_i)_++(A_i)_- \otimes (B_i)_-\right\}.
$$
\end{theorem}

We remark that the condition that $A$ and $B$ are commuting with
each other is also necessary for the equality in the theorem to
hold. The following result is quoted from standard books on Lie
algebras \cite{H}.

\begin{proposition} (Lie's theorem) Let $L$ be any solvable
subalgebra of the general linear Lie algebra, then the matrices of
$L$ relative to a suitable basis of $V$ are upper triangular.
Furthermore, one can adjust the basis to be orthonormal.
\end{proposition}

We remark that the last statement is due to the fact that the
transition matrix in Gram-Schmidt process is upper triangular. Now
we suppose that there exists a unitary matrix $Q$ such that both $A$
and $B$ are upper-triangularized as follows:
$$
A=Q\Lambda_1Q^{\dagger}, B=Q\Lambda_2Q^{\dagger},
$$
where $\Lambda_i$ are upper-triangular. If $A$ and $B$ are
hermitian, then $\Lambda_i^{\dagger}=\Lambda_i$, which forces
$\Lambda_i$ to be diagonal. Then $A$ and $B$ are actually commuting
with each other, subsequently $(A+B)_{\pm}=A_{\pm}+B_{\pm}$. So in
this context the additivity seems to be not too far away from
commutativity.

Theorem \ref{T:factor} gives the closest approximation to a
two-partite density operator by the tensor operator of non-negative
operators, however one has to fit the approximation under the
constraint of unit trace. We hope that further studies can answer
this question.

\section{Conclusion}

Matrix approximation is an old problem in mathematics with
applications in physics, finance, and computer sciences. In this
paper we have completely solved the optimization problem to
approximate any Hermitian matrix by positive semi-definite matrices.
The solution is shown to be given by the spectral decomposition of
the concerned matrix. We apply this result to density matrices and
obtain useful approximation by tensor product of density matrices
using Lie theoretic techniques. Our results also open possible deep
connection among quantum entanglement, data mining and signal
procession.

\section{Acknowledgments}

Jing gratefully acknowledges the support from NSA grant
MDA904-97-1-0062 and NSFC's Overseas Distinguished Youth Grant
(10728102).

\end{document}